\documentclass[prx,twocolumn,english,superscriptaddress,floatfix,longbibliography]{revtex4-2}
\usepackage[T1]{fontenc}
\usepackage{color}
\usepackage{babel}
\usepackage{amsthm}
\usepackage{amsmath}
\usepackage{amssymb}
\usepackage{microtype}

\usepackage{chemformula}
\usepackage[toc,page]{appendix}
\usepackage{graphicx}
\usepackage[unicode=true,pdfusetitle, bookmarks=true,bookmarksnumbered=false,bookmarksopen=false,  breaklinks=false,pdfborder={0 0 0},backref=false,colorlinks=false] {hyperref}
\hypersetup{ colorlinks,linkcolor=myurlcolor,citecolor=myurlcolor,urlcolor=myurlcolor}
\usepackage{tcolorbox}

\makeatletter
\@ifundefined{textcolor}{}
{%
 \definecolor{BLACK}{gray}{0}
 \definecolor{WHITE}{gray}{1}
 \definecolor{RED}{rgb}{1,0,0}
 \definecolor{GREEN}{rgb}{0,1,0}
 \definecolor{BLUE}{rgb}{0,0,1}
 \definecolor{CYAN}{cmyk}{1,0,0,0}
 \definecolor{MAGENTA}{cmyk}{0,1,0,0}
 \definecolor{YELLOW}{cmyk}{0,0,1,0}
 }
\theoremstyle{plain}

\theoremstyle{plain}

\ifx\proof\undefined
\newenvironment{proof}[1][\protect\proofname]{\par
\normalfont\topsep6\p@\@plus6\p@\relax
\trivlist
\itemindent\parindent
\item[\hskip\labelsep
\scshape
#1]\ignorespaces
}{%
\endtrivlist\@endpefalse
}
\providecommand{\proofname}{Proof}
\fi
\theoremstyle{plain}

\providecommand{\lemmaname}{Lemma}
\providecommand{\definitionname}{Definition}
\providecommand{\propositionname}{Proposition}

\usepackage{babel}
\usepackage{txfonts}%\usepackage{braket}
\usepackage{colortbl}\definecolor{myurlcolor}{rgb}{0,0,0.7}

\def\ket#1{| #1 \rangle}
\def\bra#1{\langle  #1 |}

\newcommand{\haH}

\newtheorem{theorem}{Theorem}

\newtheorem{definition}[theorem]{Definition}

\definecolor{orange}{RGB}{255,127,0}

\begin{document}
\preprint{APS/123-QED}
\title{Bounding Quantum Advantages in Postselected Metrology}
\author{Sourav Das}
\email{sourav.iisermohali@gmail.com}
 \affiliation{ Department of Physical Sciences, Indian Institute of Science Education and Research (IISER), Mohali, Punjab 140306, India}
\author{Subhrajit Modak}
\email{modoksuvrojit@gmail.com}
\affiliation{ Department of Physical Sciences, Indian Institute of Science Education and Research (IISER), Mohali, Punjab 140306, India}
 \author{Manabendra Nath Bera}
 \email{mnbera@gmail.com}
 \affiliation{ Department of Physical Sciences, Indian Institute of Science Education and Research (IISER), Mohali, Punjab 140306, India}

\begin{abstract}
Weak value amplification and other postselection-based metrological protocols can enhance precision while estimating small parameters, outperforming postselection-free protocols. In general, these enhancements are largely constrained because the protocols yielding higher precision are rarely obtained due to a lower probability of successful postselection. It is shown that this precision can further be improved with the help of quantum resources like entanglement and negativity in the quasiprobability distribution. However, these quantum advantages in attaining considerable success probability with large precision are bounded irrespective of any accessible quantum resources. Here we derive a bound of these advantages in postselected metrology, establishing a connection with weak value optimization where the latter can be understood in terms of geometric phase. We introduce a scheme that saturates the bound, yielding anomalously large precision. Usually, negative quasiprobabilities are considered essential in enabling postselection to increase precision beyond standard optimized values. In contrast, we prove that these advantages can indeed be achieved with positive quasiprobability distribution. We also provide an optimal metrological scheme using three level non-degenerate quantum system.

\end{abstract}
\maketitle

\section{Introduction}
Parameter estimation, which is central to mathematical statistics, is an elementary problem in information theory. Its main objective is to construct and evaluate various methods that  can  estimate  the  values of  parameters of  either an information source or a communication channel. In estimation theory, the Cram\'{e}r-Rao inequality \cite{cramer2016mathematical,rao1992information} expresses a lower bound on the variance of unbiased estimators $\theta_e$ stating that the variance of any such estimator $\text{Var}(\theta_e)$ is at least as high as the inverse of the Fisher information, $\mathcal{I}(\theta)$:
\begin{center}
$\text{Var}(\theta_e)\geq\frac{1}{\mathcal{I}(\theta)}$.
\end{center}
The  Fisher information quantifies the average information learned about an unknown parameter $\theta$ from an experiment. This is a central quantity in metrology as it sets the lower bound of the precision in estimation. A common metrological task is concerned with optimally estimating a parameter that characterizes a physical process. For that, one needs to design an experimental setup that minimizes the estimator's error by maximizing the Fisher information \cite{Giovannetti_2004}. 

It has been shown that utilization of quantum resources can improve the precision of parameter estimations beyond classical limits. This idea is at the basis of the continuously growing research area of quantum metrology that aims at reaching the fundamental bounds in metrology by exploiting quantum probes. The central quantity in quantum metrology is the quantum Fisher information (QFI) which is the optimal Fisher information of a metrological protocol over different measurement settings \cite{helstrom1976quantum,paris2009quantum,doi:10.1119/1.1308267,petz2011introduction}. Consider using $N$ classical probes, each interacting once at a time with the system under study, the estimation error will at best scale as $N^{-\frac{1}{2}}$. This limit is called Standard Quantum Limit (SQL), and it stands for the probes that are at most classically correlated. If quantum probes are allowed, the error scaling improves to $N^{-1}$, called as Heisenberg Limit (HL \cite{maccone2013intuitive,giovannetti2006metrology,zwierz2012ultimate}.  The quest for measurement schemes surpassing the SQL has inspired a variety of clever strategies, employing squeezing of the vacuum \cite{caves1981quantum,dowling2008quantum,vahlbruch2006coherent,barsotti2018squeezed,goda2008quantum}, optimizing the probing time \cite{chaves2013noisy}, monitoring the environment \cite{plenio2016sensing,albarelli2017ultimate}, and exploiting non-Markovian effects \cite{chin2012quantum,smirne2016ultimate}. Besides   the   fundamental   interest   about   ultimate   precision limits, quantum metrology presents different applications, such   as: measurement    on    biological    systems \cite{mauranyapin2017evanescent,taylor2016quantum}, gravitational waves detection \cite{abadie2011gravitational}, atomic clocks \cite{borregaard2013near,ludlow2015optical,katori2011optical}, interferometry with    atomic and molecular matter waves \cite{cronin2009optics,che2019multiqubit,tsarev2018quantum}, Hamiltonian estimation \cite{wang2017experimental,zhang2014quantum,shabani2011estimation,wiebe2014hamiltonian} and other general sensing technologies \cite{gefen2019overcoming,degen2017quantum}.  

It is considered that HL represents a fundamental limit on the sensitivity of quantum measurements. However, different studies have shown that interactions among particles may be a valuable resource for quantum metrology, allowing scaling beyond HL \cite{napolitano2011interaction,boixo2007generalized,choi2008bose} . Naturally, the question arises: What are the alternative protocols that can be relevant to achieving such scaling? A recent theoretical study has reported that the postselected quantum experiments can be used to overcome HL by enabling a quantum state to carry more Fisher information \cite{arvidsson2020quantum}. The reason behind  this benefit is claimed to be the negative quasiprobability distribution \cite{dirac1945analogy,PhysRev.44.31,PhysRevA.76.012119,PhysRevLett.101.020401,barut1957distribution,terletsky1937classical,margenau1961correlation} which is an important manifestation of nonclassicality. Also, it has been shown that improved advantages can be attained when properly conditioned experiments are performed. However, this advantage comes with a lower rate of successful postselection. 

 In this article, we go beyond this restrictive setting and ask, can this advantage be bounded fundamentally irrespective of any accessible quantum resources. We conclude that it can. Using geometric arguments, we derive a bound of these advantages using weak value optimization and show that the intrinsic weak values of the system observable play a key role in bounding this advantage. We construct a preparation-and-postselection procedure that saturates the bound using a three-level non-degenerate quantum system. Surprisingly, at the saturation point, the quasiprobability distribution of this setting turns out to be positive. So far, the negative quasiprobability distribution is considered essential for postselected QFI to overcome HL. Our scheme achieves the same without accounting for any negative or non-real elements in the quasiprobability distribution. Finally, we propose an alternative way to understand this quantum advantage.

\section{Postselected quantum Fisher information} \label{Sec:FishInf}
A typical situation in quantum parameter estimation is to estimate a parameter $\theta$ that is encoded on the system state through some general quantum evolution $\hat{\rho}_{\theta}=\mathcal{E}_{\theta}(\hat\rho)$, where $\mathcal{E}_{\theta}$ is a trace-preserving completely positive map. The parameters could be phases of light in interferometers, unitary phase shift, the decay constant of an atom, strength of magnetic or gravitational fields, etc. In order to characterize how efficient an estimation procedure is, we compare protocols based on their mean-square estimation error (MSE),
\begin{equation}
    \text{MSE}(\theta_{\text{true}})=\mathbb{E}_{\text{data}}[(\theta(\text{data})-\theta_{\text{true}})^2],
\end{equation}
where $\theta$ is an estimator for $\theta_{\text{true}}$ and $\mathbb{E}_{\text{data}}[.]$ denotes an expectation over data. The achievable precision in parameter estimation with a quantum system has lower bound on the MSE of unbiased estimators.  For quantum estimation problems, there is a strict lower bound on the MSE of unbiased estimators, known as the \emph{quantum Cram\'{e}r-Rao bound} (QCRB). The bound is expressed in terms of the quantum Fisher information $\mathcal{I}_{Q}(\theta)$ associated with the state $\hat{\rho}_{\theta}$ that encodes the parameter. For measurements on $N$ copies of the system, the QCRB is expressed as
\begin{equation}
    \text{MSE}(\theta)\geq\frac{1}{N\mathcal{I}_{Q}(\theta)}.
\end{equation}

\begin{figure}[!h]\label{mqa}
	\centering{
		\resizebox{90mm}{!}{\includegraphics{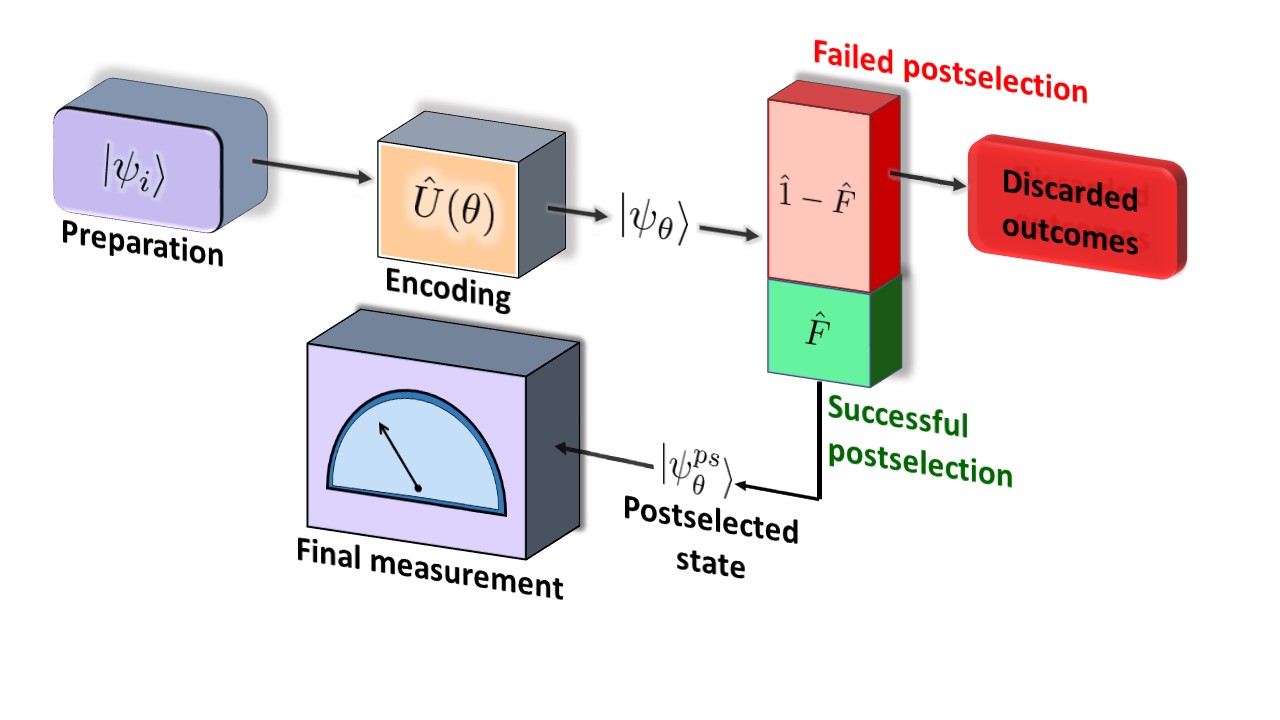}}
		\caption{A scheme for postselected metrology: First, an input quantum state $\ket{\psi_i}$ undergoes a unitary transformation $\hat{U}(\theta) = e^{-i \theta \hat{A}}$: $\ket{\psi_i} \rightarrow \ket{\psi_\theta}$. Second, the quantum state is subjected to a  postselective measurement with the projectors $\{\hat{F}, \hat{1} - \hat{F} \}$. After the successful postselection with $\hat{F}$, the  updated (renormalized)  state becomes $\ket{\psi_\theta^{ps}} = \hat{F}\ket{\psi_\theta}/(\langle \psi_\theta\vert\hat{F}\vert\psi_\theta\rangle)^{1/2}$, which is further analyzed to estimate the parameters. \color{black}}
		\label{figure:sfbac}}
\end{figure}

Consider a quantum experiment that outputs a state $\hat{\rho}_{\theta} = \hat{U}(\theta) \hat{\rho}_0 \hat{U}^{\dagger}(\theta)$, where $\hat{\rho}_0$ undergoes a unitary transformation $\hat{U}(\theta) = e^{-i \hat{A} \theta}$ driven by $\hat{A}$ with $\theta\in\mathbb{R}$. If $\hat{\rho}_{\theta}$ is pure, such that $\hat{\rho}_{\theta} =\ket{\psi_{\theta}} \bra{\psi_{\theta}}$, the optimized quantum Fisher information can be written as
\begin{equation}
\mathcal{I}_{Q}(\theta \vert \hat{\rho}_{\theta}) = 4 \textrm{Var}(\hat{A})_{\hat{\rho}_{0}}.
\end{equation}
The Fisher information can be further optimized over all input states, which gives
\begin{equation}
\max_{\hat{\rho}_0} \big\{ \mathcal{I}_Q(\theta \vert \hat{\rho}_{\theta}) \big\} = 4 \max_{\hat{\rho}_0} \big\{ \textrm{Var}(\hat{A})_{\hat{\rho}_{0}} \big\} =  (\Delta a)^2,
\end{equation}
where $\Delta a$ is the difference between the maximum and minimum eigenvalues of $\hat{A}$. In order to show how post-selection could help retrieve more information per measurement, we present a short review of the postselected prepare-measure experiment. According to the protocol, as shown in Fig.\ref{figure:sfbac}, a projective postselection takes place after $\hat{U}(\theta)$ but before the final measurement. The renormalized quantum state that passes the postselection is $\vert{\psi_{\theta}^{ps}}\rangle \equiv \vert\Psi_{\theta}^{ps}\rangle\Big/ \sqrt{p_{\theta}^{ps}}$, where the unnormalized state is defined as  $\ket{\Psi_{\theta}^{ps}} \equiv \hat{F}\ket{\psi_{\theta}} $ with  $\hat{F} = \sum_{f \in \mathcal{F}^{ps}} \ket{f} \bra{f}$ being the postselecting projection operator, $\mathcal{F}^{ps}$ is the chosen set of postselection and $p_{\theta}^{ps} \equiv \mathrm{Tr}(\hat{F}\hat{\rho}_{\theta})$ is the probability of postselection. Finally, the postselected state undergoes an information-optimal measurement, and one can obtain
\begin{equation}
\label{Eq:FishQpsShort}
\mathcal{I}_{Q}(\theta\vert\psi_{\theta}^{ps}) = 4\langle\dot{\Psi}_{\theta}^{ps}\vert \dot{\Psi}_{\theta}^{ps}\rangle\frac{1}{p_{\theta}^{ps}}-4\vert\langle\dot{\Psi}_{\theta}^{ps}\vert\Psi_{\theta}^{ps}\rangle\vert^{2}\frac{1}{(p_{\theta}^{ps})^{2}},
\end{equation}
where $\vert\dot{\Psi}_\theta^{ps}\rangle \equiv \partial_\theta \ket{\Psi_\theta^{ps}}$. 

This gives the quantum Fisher information available from a quantum state after its postselection. Evidently, $\mathcal{I}_{Q}(\theta\vert\psi_{\theta}^{ps})$ exceeds $\mathcal{I}_{Q}(\theta\vert\hat{\rho}_{\theta})$, since $p_{\theta}^{ps} \leq 1$. Moreover, if the postselections result in a quasiprobability distribution (i.e., with negative entries), the rate of improvement in retrieving information per measurement, can even surpass the standard limit of Fisher information \cite{arvidsson2020quantum}. The postselected quantum Fisher information can be re-expressed in terms of quasiprobability distribution:
\begin{align}
\mathcal{I}_Q(\theta \vert \psi_{\theta}^{ps}) =  4  \sum_{
\substack{a,a^{\prime}, \\ f\in \mathcal{F}^{ps} } 
} \frac{q_{a,a^{\prime},f}^{\hat{\rho}_{\theta}}}{p_{\theta}^{ps}}   a a^{\prime}   - 4   \Big\vert \sum_{
\substack{a,a^{\prime}, \\ f\in \mathcal{F}^{ps} } 
} \frac{q_{a,a^{\prime},f}^{\hat{\rho}_{\theta}}}{p_{\theta}^{ps}}  a \Big\vert^2
\label{ref1} ,
\end{align} 
where $q_{a,a^{\prime},f}^{\hat{\rho}_{\theta}}=\langle f\vert a\rangle \langle a\vert\hat{\rho}_{\theta}\vert a^{\prime}\rangle\langle a^{\prime}\vert f\rangle$ refers an element of doubly extended Kirkwood-Dirac (KD) quasiprobability distribution \cite{dirac1945analogy,PhysRev.44.31}, defined in terms of eigenbases of $\hat{A}$ and $\hat{F}$. This is an extension of standard KD distribution which is obtained under two weak measurements followed by a strong measurement performed sequentially onto the system. If $\hat{A}$ commutes with $\hat{F}$, as they do classically, then they share an eigenbasis for which $q_{a,a^{\prime},f}^{\hat{\rho}_{\theta}} / p_{\theta}^{ps} \in [0, \, 1] $, and the postselected quantum Fisher information is bounded as $ \mathcal{I}_Q(\theta \vert \psi^{ps}_{\theta}) \leq (\Delta a)^2$. In contrast to this, if the quasiprobability distribution contains negative values, the postselected quantum Fisher information can be seen to violate the standard bound: $\mathcal{I}_Q(\theta \vert \psi^{ps}_{\theta}) > ( \Delta a)^2$. In other words, in a classically commuting theory, no post-selected experiment can generate more Fisher information than the optimized prepare-measure experiment. We will see shortly how these non-classical advantages can be accommodated in any quantum experiments to achieve anomalously large quantum Fisher information. In fact, this gain in information can be related to the weak values of the system observable. For that, first we rewrite Eq.~(\ref{ref1}) in the operator form:
\begin{equation}\label{refq}
\begin{split}
\mathcal{I}_Q(\theta \vert \psi_{\theta}^{ps}) = \frac{4}{p_{\theta}^{ps}} \mathrm{Tr}\Big( \hat{F} \hat{A} \hat{U}(\theta) \hat{\rho}_{0} \hat{U}(\theta)^{\dagger} \hat{A} \Big) - \\
~\frac{4}{(p_{\theta}^{ps})^2} \Big\vert  \mathrm{Tr}\Big(  \hat{F} \hat{U}(\theta) \hat{\rho}_{0} \hat{U}(\theta)^{\dagger} \hat{A} \Big)\Big\vert^2 .
\end{split}
\end{equation}
The aim is to choose $\hat{F}$ and $\hat{\rho}_{0}$ in a way that $\mathcal{I}_Q(\theta \vert \psi_{\theta}^{ps})$ approaches maximum. We begin with a pure initial state $\ket{\psi_i}$ and the postselction is carried out with a  projector $\hat{F}=\sum_{{f_k}\in \mathcal{F}^{ps} }\ket{f_{k}}\bra{f_{k}}$ on the updated state $\ket{\psi_{\theta}}=\hat{U}(\theta) \ket{\psi_i}$, where $\mathcal{F}^{ps}$ is the postselection basis. Then the postselction probability becomes,
\begin{equation}
  p_{\theta}^{ps}=\sum_{{f_k}\in \mathcal{F}^{ps} }\vert\langle\psi_{\theta}\vert f_{k}\rangle\vert^{2},
\end{equation}
and Eq.~(\ref{refq}) leads to
\begin{equation}
\label{Fisheq26}
\mathcal{I}_Q(\theta\vert\psi_{\theta}^{ps})(p_{\theta}^{ps})^{2}=4\sum_{i<j}p_{\theta i}^{ps}p_{\theta j}^{ps}\Big\vert A_{w}^{i}-A_{w}^{j}\Big\vert^{2},
\end{equation}
where $A_{w}^{k}=\frac{\langle\psi_{\theta}\vert\hat{A}\vert f_{k}\rangle}{\langle\psi_{\theta}\vert f_{k}\rangle}$  and $p_{\theta k}^{ps}=\vert\langle\psi_{\theta}\vert f_{k}\rangle\vert^{2}$. It can be seen that there is no enhancement in postselected QFI when the system is postselected with rank-$1$ projector.

Note, $A_{w}^{k}$ is nothing but the weak values corresponding to the observable $\hat{A}$ between the preselected state $\ket{\psi_{\theta}}$ and the $k^{\text{th}}$component of the postselcted state $\ket{f_{k}}$, and $p_{\theta k}^{ps}$ is the postselection probability \footnote{These weak values are not directly measured through standard weak measurement. Rather, we see their effects indirectly as they play a key role in determining the postselected QFI. We shall refer to them as the intrinsic weak values of the observable} (See Supplementary Material for details). Therefore, optimizing enhancement in postselected metrology can be connected to the weak value optimization, as we discuss below.

\section{Weak value optimization: A geometric interpretation}
In general, a large weak value appears when it is less likely to have successful postselection of the system. It indicates that large weak values are obtained but very rarely. It has been shown that by exploiting quantum resources, e.g., entanglement, squeezed states, etc., the success probability can be improved for a fixed weak value. However, the advantage in attaining considerable success probability with a large weak value is bounded irrespective of any accessible quantum resources. To optimize the advantage, one can either start by fixing the weak value or keeping the success probability fixed. Usually, these are attained from separate optimization protocols. In order to access these advantages simultaneously from a single protocol, we can start with optimizing a combined quantity $\eta$: 
\begin{equation}\label{lo}
     \eta(A_w,p_s)=p_s \vert A_w \vert^2.
\end{equation}
We refer this quantity as the \emph{efficiency} of a weak value metrological protocol. To check how efficient a protocol is, one needs to quantify the gain in terms of $\vert A_w\vert$ along with a cost $\frac{1}{p_s}$ to access the same. Evidently, minimal cost implies higher probability of successful postselection. The efficiency is bounded by the following relation:
\begin{equation}
\label{weakvaluebound}
\eta(A_w,p_s)\leq\vert\vert \hat{A}^2\vert\vert_{\text{op}},
\end{equation}
where the operator norm is defined as $\vert\vert \hat{X} \vert\vert_{\text{op}}:= \sup_{\ket{\phi}\in\mathcal{H}}\{\bra{\phi}\hat{X}\ket{\phi}: \langle\phi\vert\phi\rangle=1\}$. Once the upper bound is known, optimization procedure can be initiated using a proper trial function. To make it more insightful, we would like to explore the geometric connection behind the process of optimization. A detailed analysis will soon reveal how geometric phase appears in the context of efficiency of a metrological experiment. To establish the connection, we start by taking the observable $\hat{A}=\sum_{k}a_k\ket{a_k}\bra{a_k}$. Substituting it into Eq. (\ref{lo}) leads to,
\begin{equation}\label{pp}
     \eta(A_w,p_s) =\Big\vert\sum_{k} a_{k} \vert\langle\psi_{f}\vert a_{k}\rangle\langle a_{k} \vert\psi_{i}\rangle\vert\text{exp}(i  \Phi_{g}^{a_k}(\ket{\psi_i},\ket{a_k},\ket{\psi_f}))\Big\vert^{2},
\end{equation}
where  $\Phi_{\text{g}}^{a_k}(\ket{\psi_i},\ket{a_k},\ket{\psi_f}) := \text{arg}(\langle \psi_f \vert a \rangle \langle a \vert \psi_i \rangle \langle \psi_i \vert \psi_f \rangle)$ is the Bergman angle, widely known as geometric phase \cite{MUKUNDA1993205}.  In the realm of quantum states when cyclic transition occurs starting from a preselected state $\ket{\psi_i}$ to the same state via a path that connects $\ket{a_k}$ and $\ket{\psi_f}$ through geodesic lines on the Bloch sphere, the final state acquires an excess phase over $\ket{\psi_i}$. This phase is proportional to the solid angle at the center, subtended by the geodesic triangle with vertices at $\ket{\psi_i}$, $\ket{\psi_f}$ and $\ket{a_k}$. Now, we aim to optimize the protocol for maximum efficiency. One way to ensure this is to keep the function under summation positive for all $k$. Any exception to this will not lead to the desired optimization. To saturate the bound, the postslected state is taken to be parallel  to $\hat{A}\ket{\psi_{i}}$, i.e.,
\begin{equation}
\ket{\psi_{f}}=\frac{\hat{A}\ket{\psi_{i}}}{\sqrt{\bra{\psi_{i}}\hat{A}^{2}\ket{\psi_{i}}}},
\end{equation}
which is similar to the case considered in \cite{alves2015weak}. Evaluating the expression for Bergman angle, we obtain:
\begin{equation}
\begin{aligned}
    \Phi_{g}^{a_k}(\ket{\psi_i},\ket{a_k},\ket{\psi_f})=
 \text{arg}\Big[a_k\frac{ \vert\langle\psi_i \vert a_k \rangle\vert^2}{\sqrt{\bra{\psi_{i}}\hat{A}^{2}\ket{\psi_{i}}}}\Big]+\text{arg}(\langle\psi_f \vert \psi_i \rangle).
\end{aligned}
\end{equation}
Note that, instead of $a_k$ all the terms inside the first argument function are positive, which simplifies the expression:
\begin{equation}\label{gf}
    \Phi_{g}^{a_k}(\ket{\psi_i},\ket{a_k},\ket{\psi_f})=\text{arg}(\text{sgn}(a_k))+\text{arg}(\langle\psi_f \vert \psi_i \rangle).
\end{equation}
This straightforwardly leads to,
\begin{equation}
   \textup{exp}(i\Phi_{g}^{a_k}(\ket{\psi_i},\ket{a_k},\ket{\psi_f}))=\textup{sgn}(a_k)\textup{exp}(i\phi),
\end{equation}
where $\textup{sgn}(a_k)=\frac{\vert a_k \vert}{a_k}$ is the sign function and $\phi=\text{arg}(\langle\psi_f \vert \psi_i \rangle)$  is a constant phase which does not contribute to the efficiency. The maximum efficiency is attained when $\Phi_{g}^{a_k}$ depends on the eigenvalues $a_k$ via a sign function only. This ensures that the terms under summation in Eq.~(\ref{pp}) is positive for all $k$ leading to maximum efficiency. An example using spin-1/2 system is outlined in the Supplementary Material. It is worth noting that the optimization in terms of efficiency, following the geometric argument, is more general in a sense that it takes into account both the weak values and the postselection probability.

\section{Bounding postselected metrology through weak value optimization}
In the same line, as mentioned in the context of efficiency of a weak value amplification, we can also assign a trade-off relation between the probability of postselection and the information obtained upon postselection. We begin with defining a similar quantity as metrological efficiency here as well.
\begin{definition}
For any arbitrary state preparation and postselection, the efficiency of a postselected metrological protocol $\xi^{ps}$ with postselected Fisher
information $\mathcal{I}_{Q}$ and total probability of successful postselection $p_{\theta}^{ps}$ is defined as,
\begin{equation}
    \xi^{ps}(p_{\theta}^{ps};\mathcal{I}_{Q})=p_{\theta}^{ps}\mathcal{I}_Q.
\end{equation}
\end{definition}  

This efficiency cannot be arbitrarily large, as mentioned in the theorem below.
\begin{theorem}
The efficiency of a protocol
for postselected metrology is bounded according to the following inequality:
\begin{equation}
\label{theorem2}
0\leq \xi^{ps}(p_{\theta}^{ps};\mathcal{I}_{Q})\leq4\vert\vert \hat{A}^2\vert\vert_{\text{op}}
\end{equation}
\end{theorem}

\begin{proof}
Since the Fisher information decreases when the states are mixed, we can expect that the maximum efficiency will be achieved for pure states only. Here, we start with preparing the system in a pure initial state $\ket{\psi_i}$ which evolves to $\ket{\psi_\theta}$ after the parameter $\theta$ is encoded via unitary $\text{exp}(-i\hat{A}\theta)$. Then, the efficiency of the protocol can be expressed in terms of standard KD distribution and efficiency in weak value amplification,
\begin{equation}
    \xi^{ps}(p_{\theta}^{ps};\mathcal{I}_{Q})=4\sum_{{f_k}\in \mathcal{F}^{ps} } \eta(p_{\theta k}^{ps}, A_w^k)-\frac{4}{p_{\theta}^{ps}}\Big\vert\sum_{m}\sum_{{f_k}\in \mathcal{F}^{ps} }
    a_m q^{\ket{\psi_\theta}}_{a_m,f_k}\Big\vert^2,
    \label{alterEf}
\end{equation}
where $q^{\ket{\psi_\theta}}_{a_m,f_k}:= \langle \psi_\theta \vert a_m \rangle \langle a_m\vert f_k \rangle \langle f_k \vert \psi_\theta \rangle$ is the standard KD distribution (see Supplementary Material for derivation).
To maximize the efficiency we shall treat the first and second term in RHS of this expression independently.  Applying Bessel's inequality for the first term, we obtain 
\begin{equation}
\label{fisheq20}
    \begin{split}
        \sum_{{f_k}\in \mathcal{F}^{ps} } \eta(p_{\theta k}^{ps}, A_w^k) = \sum_{{f_k}\in \mathcal{F}^{ps} } \vert \langle \psi_\theta \vert \hat{A}\vert f_k \rangle\vert^2 \leq \vert\vert  \hat{A}^2 \vert\vert_\text{op} .
    \end{split}
\end{equation}
Since the second term is a non-negative quantity, the inequality (\ref{theorem2}) is always respected.

\end{proof}
Now we investigate the condition to achieve optimal efficiency.
The inequality (\ref{fisheq20}) saturates under the following condition,
\begin{equation}
    \sum_{{f_k}\notin \mathcal{F}^{ps} } \eta(p_{\theta k}^{ps}, A_w^k) = 0 .
\end{equation}
Since all the terms inside the summation sign is positive, the condition is equivalent with $\forall f_k \notin \mathcal{F}^{ps}, \ \langle \psi_\theta \vert \hat{A}\vert f_k\rangle =0$. This implies that all the intrinsic weak values with failed postselected states have to be zero separately. The condition for the second term in RHS of Eq.~(\ref{alterEf}) to be zero is  $\sum_{m}\sum_{{f_k}\in \mathcal{F}^{ps} } a_m q^{\ket{\psi_\theta}}_{a_m,f_k}=0$. Putting together both the conditions, we obtain 
    \begin{equation}
        \begin{split}
            \langle \psi_i \vert \hat{A} \vert \psi_i \rangle = 0 \; \; \text{and} \; \; \forall f_k \notin \mathcal{F}^{ps}, \ \ A_w^k =0 .
        \end{split}
    \end{equation}
Under this condition, the inequality (\ref{theorem2}) saturates, which is termed as information preserving postselected metrology.
    
Now, we point out a situation where we achieve the bound for a rank-2 postselection. Then, Eq.~(\ref{Fisheq26}) reduces to $\mathcal{I}_Q(\theta\vert\psi_{\theta}^{ps})(p_{\theta}^{ps})^{2}=4 p_{\theta 1}^{ps}p_{\theta 2}^{ps}\Big\vert A_{w}^{1}-A_{w}^{2}\Big\vert^{2}$. First we set $A_{w}^{2}=0$. Under this condition, we will get:
\begin{equation}
\xi^{ps}(p_{\theta}^{ps};\mathcal{I}_{Q})=4\frac{p_{\theta 1}^{ps}p_{\theta 2}^{ps}}{p_{\theta 1}^{ps}+p_{\theta 2}^{ps}}\vert A_w^1\vert^{2}.
\end{equation}
 Imposing the weak value optimization (\ref{weakvaluebound}) on first postselection leads to,
\begin{equation}
\xi^{ps}(p_{\theta}^{ps};\mathcal{I}_{Q})=4\frac{p_{\theta 2}^{ps}}{p_{\theta 1}^{ps}+p_{\theta 2}^{ps}}\vert\vert \hat{A}^{2}\vert\vert_{\text{op}}.
\end{equation}
Once $p_{\theta 1}^{ps}$ starts to tend towards zero, one can get closer to the saturation. Thus,  we need to maintain: $\vert f_{2}\rangle$ perpendicular to $\hat{A}\vert\psi_{\theta}\rangle$ and $\vert\psi_{\theta}\rangle$ parallel to $\hat{A}\vert f_{1}\rangle$, simultaneously. Also, we need to choose the initial state of the system in such a way that the average of $\hat{A}^2$ saturates under the optimality conditions. It is now evident from the geometric interpretation that the optimization for the first intrinsic weak value would be attained when the geometric phases between the preselection, the eigenstates of $\hat{A}$ and the postselection have certain dependence with the eigenvalues of $\hat{A}$.  Below, give an example to show how to achieve optimal quantum enhancement in the context of postselected metrology. 

\subsection{Information preserving postselected metrology using three-level quantum system}

\begin{figure*}[htb!]
	\minipage{0.28\textwidth}%
	\includegraphics[width=\linewidth,height=4.5cm]{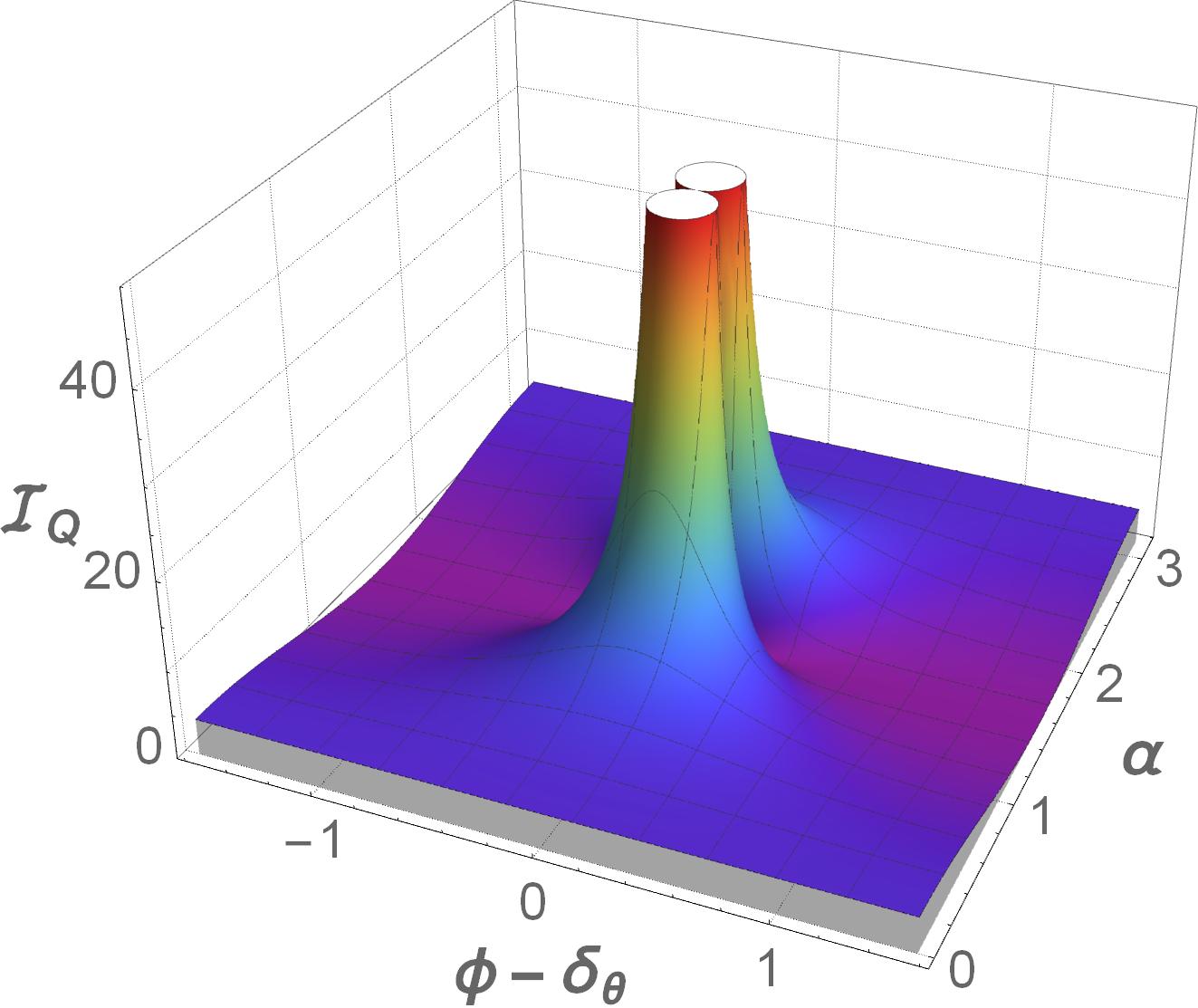}
	\endminipage\hfill
	\minipage{0.28\textwidth}
	\includegraphics[width=\linewidth]{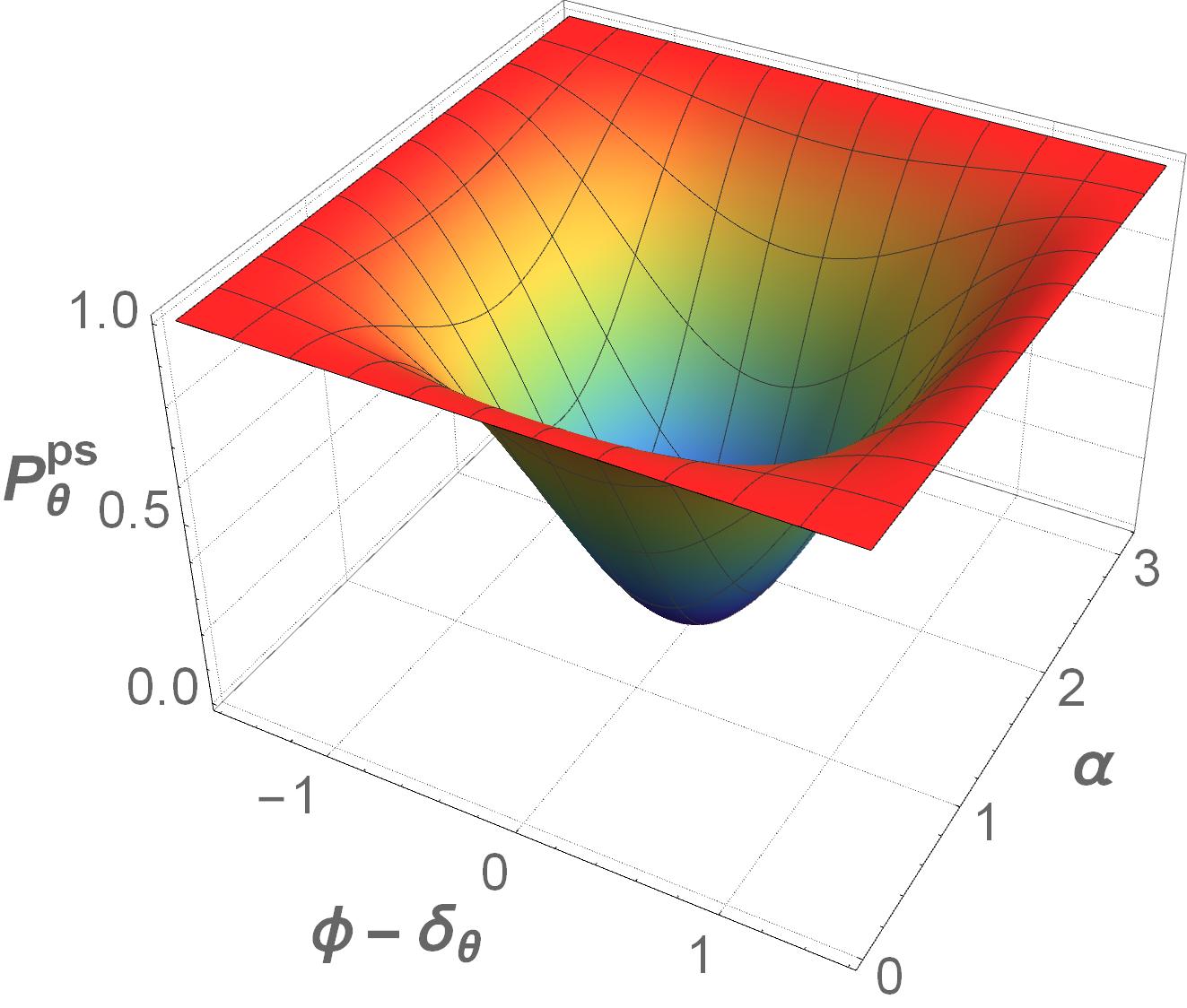}
	\endminipage\hfill
	\minipage{0.28\textwidth}
	\includegraphics[width=\linewidth]{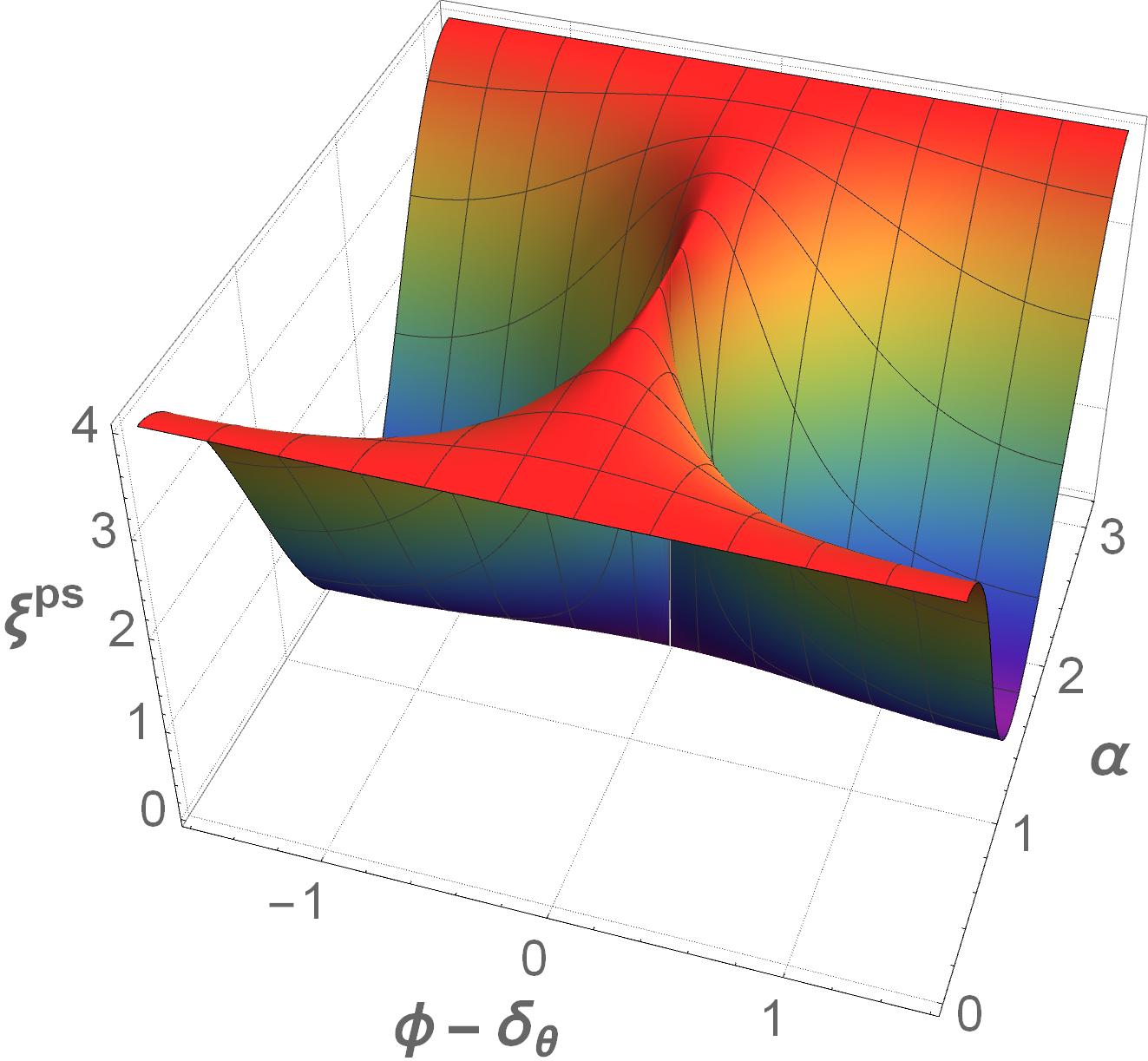}
	\endminipage
	\caption{\textbf{Information preserving postselected metrology.} Figures from the left represent the postselecetd Fisher information $\mathcal{I}_Q$, probability of successful postselection $p_{\theta}^{ps}$, and efficiency of the protocol respectively with different values of $\phi-\delta_\theta$ and $\alpha$, for $\lambda=1$. Optimatility condition is attained when  $\phi \rightarrow \delta_\theta$. For more details, see text.}    
	\label{Fisherinfofig}
\end{figure*}

Consider the state space of the system is spanned by the basis set $\{\ket\lambda,\ket{\tilde\lambda}$, $\ket{-\lambda}\}$, which are the eigenvectors of the system observable $\hat{A}$ with corresponding  eigenvalues $\lambda$, $\tilde{\lambda}$ and $-\lambda$ respectively. We want to estimate $\theta$ that lies close to its true value $\theta_0$ and the difference between them is $\delta_\theta \equiv \theta - \theta_0$, with $\vert\delta_{\theta}\vert\ll1$. First, we choose the postselection $\hat{F}=\ket {f_{1}}\bra {f_{1}}+\ket {f_{2}}\bra {f_{2}}$, with
\begin{equation}
\vert f_{1}\rangle=\frac{\ket{\lambda}-\ket{-\lambda}}{\sqrt{2}},
\end{equation}
\begin{equation}
\vert f_{2}\rangle=\frac{\text{cos}\alpha\vert\lambda\rangle+\text{cos}\alpha\ket{-\lambda}}{\sqrt{2}}+\text{sin}\alpha\ket{\tilde\lambda},
\end{equation}
where $\alpha$ is a real parameter. We also choose the initial state, 
\begin{equation}
    \ket{\psi_{i}}=\frac{U^{\dagger}(\theta_0)\Big[(\text{cos}\phi+i\text{sin}\phi)\ket{\lambda}+(\text{cos}\phi-i\text{sin}\phi)\ket{-\lambda}\Big]}{\sqrt{2}},
\end{equation}
where the unitary $\hat{U}(\theta_0) = \text{exp}(-i \hat{A} \theta_0)$. 
 In this setting the evolved state becomes,
\begin{equation}
\begin{aligned}
\ket{\psi_{\theta}}=\hat{U}(\theta)\ket{\psi_{i}}=\frac{1}{\sqrt{2}}\Big[(\text{cos}\phi+i\text{sin}\phi)e^{-i\delta_{\theta}\lambda}\ket{\lambda}+\\
~(\text{cos}\phi-i\text{sin}\phi)e^{i\delta_{\theta}\lambda}\ket{-\lambda}\Big],
\end{aligned}
\end{equation}
where $\phi$ is a real parameter. Putting this postselections and the initial state into Eq.~(\ref{Fisheq26}), we find
\begin{equation}
\mathcal{I}_Q(\theta\vert\psi_{\theta}^{ps})p_{\theta}^{ps}=\frac{4\lambda^2\text{cos}^2\alpha}{\text{sin}^2(\phi-\lambda\delta_{\theta})+\text{cos}^2\alpha\text{cos}^2(\phi-\lambda\delta_{\theta})}.
\end{equation}
The total postselection probability comes out of this estimation:
\begin{equation}
    %\begin{align}
        p_{\theta}^{ps}=p_{\theta 1}^{ps}+p_{\theta 2}^{ps}
        ~=\text{sin}^2(\phi-\lambda\delta_{\theta})+\text{cos}^2\alpha\text{cos}^2(\phi-\lambda\delta_{\theta}),
   % \end{align}
\end{equation}
along with the respective weak values
\begin{equation}
    A_{w}^{1}=\frac{-i\lambda\text{cos}(\phi-\lambda\delta_{\theta})}{\text{sin}(\phi-\lambda\delta_{\theta})},   ~~A_{w}^{2}=\frac{i\lambda\text{sin}(\phi-\lambda\delta_{\theta})}{\text{cos}(\phi-\lambda\delta_{\theta})}.
\end{equation}

As $\phi$ approaches $\lambda\delta_\theta$,  the condition for optimal quantum advantages gets
 satisfied. This can clearly be seen from the following expressions when compiling with appropriate limits,
\begin{align}
\lim_{  \phi \rightarrow \lambda\delta_{\theta} }   p^{ps}_{\theta}  & = \text{cos}^2\alpha  , \\
\lim_{  \phi \rightarrow \lambda\delta_{\theta} }  \mathcal{I}_Q(\theta \vert \psi_{\theta}^{ps})& = 4\lambda^2\text{sec}^2\alpha  , \; \; \textrm{and} \label{Eq:ArbFishPS} 
 \\
\lim_{  \phi \rightarrow \lambda\delta_{\theta} } \xi (p^{ps}_{\theta}; \mathcal{I}_Q) & = 4\lambda^2. 
\end{align}
Thus, in the aforementioned limit, we achieve the information preserving protocol associated with the postselected metrology. The Fig.~\ref{Fisherinfofig} reflects how the Fisher information, probability of postselection and the efficiency of the protocol are dependent on the change of parameters $\{\phi -\lambda\delta_{\theta}, \alpha\}$. The parameter $\alpha$ opens up a choice for the experimenter to decide the degree of precision without violating the restriction of information preservation. Before moving to further discussion, we should note that tuning $\phi$ in the limit $\phi \rightarrow \lambda\delta_{\theta}$ would not be an easy task. We must have a pre-estimated error range of the parameter $\theta$ beforehand; let us call it $\Delta_\theta$.  The experimenter then needs to perform a trial run tuning $\phi$ with different values within the range $\vert \phi\vert \leq \lambda \Delta_\theta$. When $\Delta_\theta$ is sufficiently small, reaching the optimal limit becomes easier. Even if the tuning is not perfect for reaching the optimal point, it gives us a range of different values of $\phi$ where we may experience anomalous QFI (see Fig.~\ref{Fisherinfofig}). To widen this range, one can consider increasing the value of $\lambda$.

In comparison to the set up \cite{arvidsson2020quantum}, our protocol relaxes some experimental restrictions yet achieves anomalous QFI without any loss of information. Additionally, we can experience this anomalousness in this setup even without setting $\delta_\theta \rightarrow 0$ condition.  Moreover, all these advantages can be obtained only by considering a three-level system. Thus, by construction, this setup is more profitable to attain the desired results.\par

The preparation and final measurement in an actual experiment have costs. Now the question arises, when the costs corresponding to each measurement set up are included, whether it still be possible to attain the improved rate of Fisher information. The answer is yes. Apart from the cost for preparation, if our experimental conditions the execution of the final measurement on successful postselection of a fraction $p_{\theta}^{ps}$ of the states, we need to include a cost for postselection. As defined in \cite{arvidsson2020quantum}, the postselected experiment’s information-cost rate is $R^{ps}(\theta):=p_{\theta}^{ps}\mathcal{I}^{ps}(\theta)/(C_P+p_{\theta}^{ps} C_M+C_{ps})$, where $\mathcal{I}^{\text{ps}}(\theta)$ is the Fisher information conditioned on successful postselection. So far, we have shown that $\mathcal{I}_{Q}(\theta\vert\psi_{\theta}^{ps})$ can exceed $4\lambda^2$. But how large can $\mathcal{I}_{Q}(\theta\vert\psi_{\theta}^{ps})$ grow?  For that, we
construct a properly conditioned metrological procedure using only a three-level non-degenerate quantum system. If $C_P$ and $C_{ps}$ are negligible compared to $C_M$, then $R^{ps}(\theta)$ can grow without any theoretical bound. In general, when $\mathcal{I}_{Q}(\theta\vert\psi_{\theta}^{ps})\rightarrow\infty$, $\xi(p_{\theta}^{ps}, \mathcal{I_Q})<4\lambda^2$, information is lost in the events discarded by postselection. We derived the proper condition to ensure no information is stored in the discarded event. These conditions in terms of the intrinsic weak values allowed us to construct this example to achieve this optimum.

\subsection{Role of quasiprobalility in postselected metrolgy}
The reason that is pointed out for having anomalously large QFI in postselected metrology is the negativity of KD distribution \cite{arvidsson2020quantum}. On the contrary, we claim that anomalously large QFI can be achieved with positive KD distribution. 

We aim to prove this claim by the same example mentioned earlier. The KD distribution corresponding to the system state is positive everywhere when the optimality conditions are imposed; see Table~\ref{quasitable1}. However, under these conditions, we obtain anomalously large QFI. 
\begin{center}
    \begin{table}[ht]
        \centering
        \begin{tabular}{c|ccc}
\textrm{$\lim_{  \phi \rightarrow \lambda\delta_{\theta} }$} {$ q^{\hat{\rho}_\theta}_{a_m, f_k}$} &
\textrm{$\ket{f_1}$} &
\textrm{$\ket{f_2}$}&
\textrm{$\ket{f_3}$}\\
\hline \hline 
$\ket{\lambda}$  & $0$ & $\frac{\text{cos}^2\alpha}{2}$ & $\frac{\text{sin}^2\alpha}{2}$ \\
\\
$\ket{\tilde\lambda}$  & $0$ & $0$ & $0$ \\
 \\
$\ket{-\lambda}$  & $0$ & $\frac{\text{cos}^2\alpha}{2}$ & $\frac{\text{sin}^2\alpha}{2}$

\end{tabular}
        \caption{KD distribution corresponding to the system state is tabulated. Under optimality condition, the distribution is positive for all values of $\alpha$. Here, $\{\ket{\lambda},\ket{\tilde{\lambda}},\ket{-\lambda}\}$ are the eigenvectors of system observable $\hat{A}$ and $\{\ket{f_1},\ket{f_2},\ket{f_3}\}$ is the complete set of postselection bases.} 
        \label{quasitable1}
    \end{table}
\end{center}

It can also be shown that the doubly extended KD distribution for this case is also positive (see Supplementary Material). Here, we propose an alternative approach for pure states. We start by referring  an identity that addresses the relation between KD distribution and the classical joint probability distribution for two successive projective measurements  \cite{PhysRevA.76.012119}: 
\begin{equation}
     q^{\hat{\rho}}_{a_m, f_k} =
    \mathrm{Tr}(\hat{\rho} \hat{A}_m \hat{F}_k \hat{A}_m) +
    \frac{1}{2} \Big[\mathrm{Tr} ( (  \hat{\rho} - \hat{\rho}'
     ) \hat{F}_k) +\mathrm{Tr} ((  \hat{\rho} - \hat{\rho}'
     ) \hat{F}_k^{\frac{\pi}{2}})\Big],
    \label{mq}
 \end{equation}
 where $\hat{F}_k^{\frac{\pi}{2}}=\text{exp}(-i\frac{\pi}{2}\hat{A}_m)\hat{F}_k\text{exp}(i\frac{\pi}{2}\hat{A}_m)$ and $\hat{\rho}'$ is the post-measurement state achieved after a non-selective projective measurement on $\hat{\rho}$ with $\hat{A}_m$ and $(\hat{1}-\hat{A}_m)$. $\{\hat{A}_{m}\}$ and $\{\hat{F}_k\}$ are two different orthonormal sets of projectors. Here, the KD distribution is composed of two terms. First term in RHS is called as Wigner formula, which refers the classical joint probability distribution for a set of outcomes $a_m$ and $f_k$ upon a successive projective measurement on $\hat{\rho}$ first with $\hat{A}_m$ and then with $\hat{F_k}$. The second and third terms together can be considered as quantum modification terms which arise due to the the non-commutativity of measurement observables. Whenever the observables commute, this part vanishes. For simplicity, we denote the Wigner formula as $Q^{\hat{\rho}}_{a_m,f_k} \; := \; \mathrm{Tr}(\hat{\rho} \hat{A}_m \hat{F}_k \hat{A}_m)$.

\begin{theorem}
For pure states, the metrological efficiency $ \xi^{ps}(p_{\theta}^{ps};\mathcal{I}_{Q})=0$, when its KD distribution equals to the corresponding Wigner formula,
\begin{equation}
q^{\hat{\rho}_\theta}_{a_m, f_k}= Q^{\hat{\rho}_\theta}_{a_m,f_k}
\end{equation}
\end{theorem}
\begin{proof}
Upon solving to satisfy the above condition, one can straightforwardly obtain two conditions:
\begin{equation}
\label{eq47}
        \langle\psi_{\theta}\vert a _m\rangle\langle a_{m}\vert f_{k}\rangle=0 
\end{equation}
\hspace{4cm} or,
\begin{equation}
\label{eq48}
        \langle a_m\vert\psi_{\theta}\rangle\langle f_k\vert a_m\rangle=\langle f_k\vert\psi_{\theta}\rangle.
\end{equation}
The condition (\ref{eq47}) refers to the points where the KD distribution and the Wigner formula both take zero values. To keep the triviality aside, we neglect that and consider the condition (\ref{eq48}).
 Recall Eq.~(\ref{alterEf}),
\begin{equation}
     \xi^{ps}(p_{\theta}^{ps};\mathcal{I}_{Q})=4\sum_{{f_k}\in \mathcal{F}^{ps} } p_{\theta k}^{ps}\vert A_w^k\vert^2-\frac{4}{p_{\theta}^{ps}}\Big\vert\sum_{m}\sum_{{f_k}\in \mathcal{F}^{ps} }
    a_m q^{\ket{\psi_\theta}}_{a_m,f_k}\Big\vert^2.
    \label{vv}
\end{equation}
Decomposing the terms under summation and imposing the condition (\ref{eq48}) lead to:
\begin{align*} 
\sum_{f\in \mathcal{F}^{ps} } p_{\theta k}^{ps}\vert A_k\vert^2=p_{\theta}^{ps} \Big\vert\sum_{m}a_m\Big\vert^2,\\ 
\sum_{m}\sum_{{f_k}\in \mathcal{F}^{ps} }
    a_m q^{\ket{\psi_\theta}}_{a_m,f_k}=  (p_{\theta}^{ps})^2 \Big\vert\sum_{m}a_m\Big\vert^2.
\end{align*}
 Replacing this to Eq.~(\ref{vv}) yields:
\begin{equation}
\xi^{ps}(p_{\theta}^{ps};\mathcal{I}_{Q})=0 . 
\end{equation}
See Supplementary Material for detailed calculation.
\end{proof}

We have emphasized this point earlier that KD distribution and the Wigner formula are not equal when the measurement observables do not commute. However, the KD distribution may still be positive in such cases. The example here also exhibits a similar feature (see Table~\ref{quasitable2}), where the Wigner formula differs from KD distribution. This observation provides a newer insight which indicates that the standard notion of KD nonclassicality \footnote{Standard notion of KD nonclassicality refers negative or non-real entries in KD distribution.} and the quantum advantage in postselected metrology may not have a direct connection.
Rather the quantum advantage may have links with the nonvanishing quantum modification terms in KD distribution.

\begin{table}[ht]
\centering
\begin{tabular}{c|ccc}
\textrm{$\lim_{  \phi \rightarrow \lambda\delta_{\theta} }$} {$ Q^{\hat{\rho}_\theta}_{a_m, f_k}$} &
\textrm{$\ket{f_1}$} &
\textrm{$\ket{f_2}$}&
\textrm{$\ket{f_3}$}\\
\hline \hline 
$\ket{\lambda}$  & $\frac{1}{4}$ & $\frac{\text{cos}^2\alpha}{4}$ & $\frac{\text{sin}^2\alpha}{4}$ \\
\\
$\ket{\tilde\lambda}$  & $0$ & $0$ & $0$ \\
 \\
$\ket{-\lambda}$  & $\frac{1}{4}$ & $\frac{\text{cos}^2\alpha}{4}$ & $\frac{\text{sin}^2\alpha}{4}$
\end{tabular}
\caption{Wigner formula corresponding to the system state.}
\label{quasitable2}
\end{table}

Thus, we have shown that the quantum advantage in postselected metrology does not appropriately reciprocate with the standard KD nonclassicality, which is often connected with negative or non-real values in the distribution. Contrary to this, we have shown there exist states with positive KD distribution, which can lead the experimenter to have QFI beyond HL in postselected metrology. Our results indicate that the quantum modification terms in the identity  Eq.~(\ref{mq})  effectively contribute to the efficiency in postselected metrology, which in consequence shows anomalous QFI.

\section{Conclusion}
We have studied postselected quantum metrology that has been shown to have better precision over other metrological protocols. It can even yield infinite precision, however, with very low probability. In a realistic situation, one has to consider both precision as well as probability. So far, an optimization protocol, which genuinely signifies the quantum advantages considering both of these quantities,  was missing. 

%The protocol we have considered here is different from weak value amplification. In the latter case, the system is weakly coupled with an external pointer followed by a postselection on the system and the final estimation is then performed on the pointer. Whereas in postselected metrology no such pointer is involved.

We have shown that the accessible advantage is bounded irrespective of the quantum resources utilized in postselected metrology. Our results highlight the significance of weak values to study the quantum advantage. In general, the weak values and the geometric phases are interrelated \cite{Sj_qvist_2006, Tamate_2009}. We have provided a unique way to examine the optimality conditions by analyzing the geometric phases associated with the states under consideration. Translating this scheme to postselected metrology has found an essential geometric argument in finding the optimal way to engineer the postselection. 

Generally, the postselected states carrying anomalously large quantum Fisher information, i.e., high precision, is attributed to negative entries in the corresponding Kirkwood-Dirac quasiprobability distribution. On the contrary, we have shown that this negativity is not necessarily associated with the increase in Fisher information. For that, we have proposed a preparation and postselection procedure using a three-level non-degenerate quantum system. 

In summary, our work brings a better understanding of postselected metrology and provides a fundamental bound of quantum advantages irrespective of the quantum resources utilized. We provide a metrological protocol that can harness the maximum possible quantum advantage. Our results lay a foundation for further exploration of multiparameter postselected metrology and possible quantum advantages. Our study indicates that the negativity of Kirkwood-Dirac distribution does not correspond to nonclassicality always. Thus, it is expected to initiate an effort for a complete characterization of nonclassicality on the theoretical ground.

{\bf Acknowledgments} --  M.N.B. gratefully acknowledges financial supports from SERB-DST (CRG/2019/002199), Government of India.

{\bf Author contributions} -- All the authors have contributed
to this work equally.

{\bf Competing interests} -- The authors declare no competing
interests.

{\bf Data and materials availability} -- Data sharing not applicable
to this article, as no datasets were generated or analyzed
during the current study.

\bibliography{apssamp1}

\newpage

\onecolumngrid

\newpage
\clearpage
\appendix

\section{Weak value amplification}
In standard weak value amplification (WVA) protocol, an experimenter prepares the meter and the system in some pure initial states $\vert\phi\rangle$
and $\vert\psi_{i}\rangle$ respectively. Then they are coupled weakly using the interaction Hamiltonian:  $\hat{H}_{\text{int}}=\hbar g\hat{A}\otimes\hat{F}\delta(t-t_{0})$, where $\hat{F}$ and $\hat{A}$ are observables corresponding to the meter and system respectively, and $g$ is a small parameter signifying the coupling strength between the system and meter. The function $\delta(t-t_0)$ indicates that the interaction between the system and the meter is impulsive. Finally, the postselection of the system is done onto a pure final state $\vert\psi_{f}\rangle$, discarding all the other events where the postselection fails. This procedure effectively prepares an updated meter state that includes the effect of the system $\vert\phi'\rangle=\hat{M}\vert\phi\rangle/\vert\vert\hat{M}\vert\phi\rangle\vert\vert$, which is mentioned in terms of a Kraus operator $\hat{M}=\langle\psi_{f}\vert\exp(-i g\hat{A}\otimes\hat{F})\vert\psi_{i}\rangle$.
Averaging a meter observable $\hat{R}$ using this updated meter state
yields $\langle\hat{R}\rangle_{\vert\phi^\prime\rangle}=\langle\phi \vert\hat{M}^{\dagger}\hat{R}\hat{M}\vert\phi\rangle/\langle\phi \vert\hat{M}^{\dagger}\hat{M}\vert\phi\rangle$. The observable average is well approximated up to first order in~$g$ \cite{kofman2012nonperturbative,di2012full}:
\begin{equation}
\langle\hat{R}\rangle_{\vert\phi'\rangle}\approx2g\left[\text{Re}A_{w}\,\text{Im}\alpha+\text{Im}A_{w}\,\text{Re}\alpha\right],\label{eq:linear}
\end{equation}
where $\alpha=\langle\hat{R}\hat{F}\rangle_{\vert\phi\rangle}$ is the correlation
parameter that can be fixed by the choice of meter observables and the initial meter
state $\vert\phi\rangle$, and $A_{w}=\langle\psi_{f}\vert\hat{A}\vert\psi_{i}\rangle/\langle\psi_{f}\vert\psi_{i}\rangle$ is a complex \emph{weak value} controlled by the system.  This relation shows how a large weak value can enhance the sensitivity of the meter even with small $g$. In case of estimating $g$, the weak value has notable ability to amplify its effect in the meter. 

\section{Generalized protocol for optimization} 
All the optimization discussed so far can be generalized in terms of efficiency.
At the point of optimization, the weak value and the postselection probability turns out to be 
\begin{equation}\label{gg}
  p_{s}=\frac{\langle\psi_{i}\vert\hat{A}\vert\psi_{i}\rangle^{2}}{\langle\psi_{i}\vert \hat{A}^{2} \vert\psi_{i}\rangle}, ~~A_{w}=\frac{\langle\psi_{i}\vert \hat{A}^{2}\vert\psi_{i}\rangle}{\langle\psi_{i}\vert \hat{A}\vert\psi_{i}\rangle}.
\end{equation}
Interesting to note that only the real part of the weak value will contribute when the point of saturation is reached. At this point, anomalously large weak values can be obtained when the observable average of $\hat{A}$ tends to zero.  It has been shown that for a fixed weak value, the probability of successful postselection can be maximized by taking the final state parallel to  $(\hat{A}-A
_w)\ket{\psi_i}$ \cite{PhysRevLett.113.030401}. Alternatively, one can also conduct optimization to maximize the weak value for a fixed postselection probability. To accomplish such enhancements, one needs to use different optimization protocols. If we optimize the efficiency $\eta(A_w,p_s)$, it is still possible to attain the optimized values of the aforementioned protocols as special cases. Here, we demonstrate this conjunction of two different protocols with an example. 

We consider coupling $n$ entangled systems to the meter simultaneously where the system qubits are prepared in an entangled state \cite{giovannetti2011nat},
\begin{equation}
\vert\psi_{i}\rangle=\frac{1}{\sqrt{2}}(\vert\lambda_x\rangle^{\otimes n}+\vert\lambda_y\rangle^{\otimes n}).
\end{equation}
Here, $\lambda_x$ and $\lambda_y$ are two arbitrary eigenvalues of an operator $\hat{a}$ whose eigenvectors spans a subsystem. The rest of the subsystems are different copies of the same. From this preparation we construct the overall system observable as follows,
\begin{equation}
    \hat{A} = \sum_{k=1}^{n} \mathbf{I}^{\otimes k-1}\otimes \hat{a}\otimes \mathbf{I}^{\otimes n-k},
\end{equation}
where the operators in the tensor product act sequentially with different subsystems. The weak value, following Eq. (\ref{gg}), is
\begin{equation}
  A_{w}=\frac{n(\lambda_x^{2}+\lambda_y^{2})}{(\lambda_x+\lambda_y)} .
\end{equation}
Under the condition for fixed weak value, we simply obtain
\begin{equation}
    \lambda_y=\frac{A_{w}\pm\sqrt{A_{w}^{2}-4n^{2}\lambda_x^{2}+4n A_{w}\lambda_x}}{2n}.
\end{equation}
For large weak values the above expression can be approximated as $\lambda_{y}\approx \{(A_{w}-n\lambda_x)/2n,-\lambda_x\}$. The first root of $\lambda_y$ corresponds to the case of non-anomalous weak value amplification. Under this choice, the eigenvalues of the operator are of the order of weak value. However, the second choice, i.e., $\lambda_y\approx-\lambda_x$, corresponds to anomalous weak value amplification where the observable average with respect to the initial state approaches zero. The postselection probability for the optimal weak value turns out to be,
\begin{equation}
 p_{s}=n^{2}\lambda_x^{2}/A_{w}^{2}\label{de}.
\end{equation}

\begin{figure}[!ht]
        \centering{
            \resizebox{75mm}{!}{\includegraphics{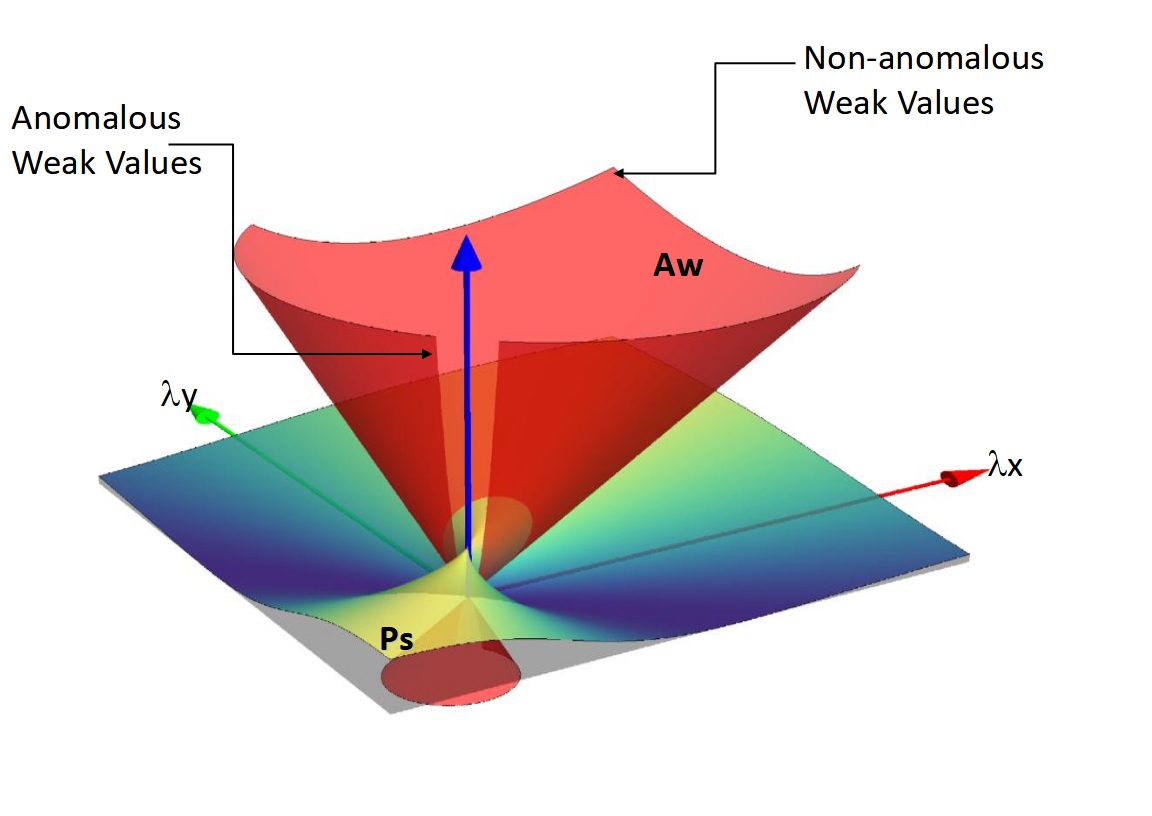}}
           \caption{Plot for optimal weak values and postselection probability for different initial state preparations characterized
by $\lambda_x$ and $\lambda_y$. The red surface corresponds to the optimal weak values $A_w$. The steep regions in this surface correspond to anomalous weak values, where they change rapidly upon small changes in the eigenvalues $\lambda_x$ and $\lambda_y$ due to anomalous amplification. In the non-anomalous region, the gradient is relatively low since the weak values scale with the eigenvalues here. The surface below stands for the corresponding success probability $p_s$.Blue regions stand for low, and the yellow regions stand for high success probability.  \color{black}}
\label{wkvalueimage}}
\end{figure}

Thus, the probability scales quadratically with $n$ and shows improvement in estimation while comparing with the classical cases.  Now, we will explore the exact opposite case where the postselction probability is held fixed with a varying $A_w$. Starting with the same initial entangled state and using Eq. (\ref{gg}), we arrive directly at
\begin{equation}
   \lambda_y=\frac{2\lambda_x\pm\sqrt{4\lambda_x^2-4(2p_{s}-1)^2\lambda_x^2}}{4p_{s}-2}.
\end{equation}

For very low $p_{s}$, this expression can be approximated as $\lambda_y\approx-\lambda_x$. Interestingly, we have traced the same approximated condition where anomalously large weak values appear. This yields the optimal weak value for a given probability of postselection
\begin{equation}\label{cd}
    \vert A_w\vert=n\lambda_x/\sqrt{p_{s}},
\end{equation}
scales linearly with $n$.  Clearly, the optimal values of $A_w$ and $p_s$ are appearing from a unified optimization process. In experiments, this unified optimization scheme would access large weak values at the cost of low probability and vice versa if the situation arises. Besides, this scheme also incorporates the cases where large eigenvalues of the system observable mostly do the amplification. This situation is referred to as the non-anomalous regime. In this regime, the weak values appear at the level of average values of the observable without any anomalous amplification. This would help us visualize the transition from weak values to average values, which is caused by the choice of preselection and postselection of the system in this scenario.

\section{Example demonstrating the geometric optimization of WVA}
\begin{figure}[ht]
	\centering
	\includegraphics[width=6cm,height=6cm]{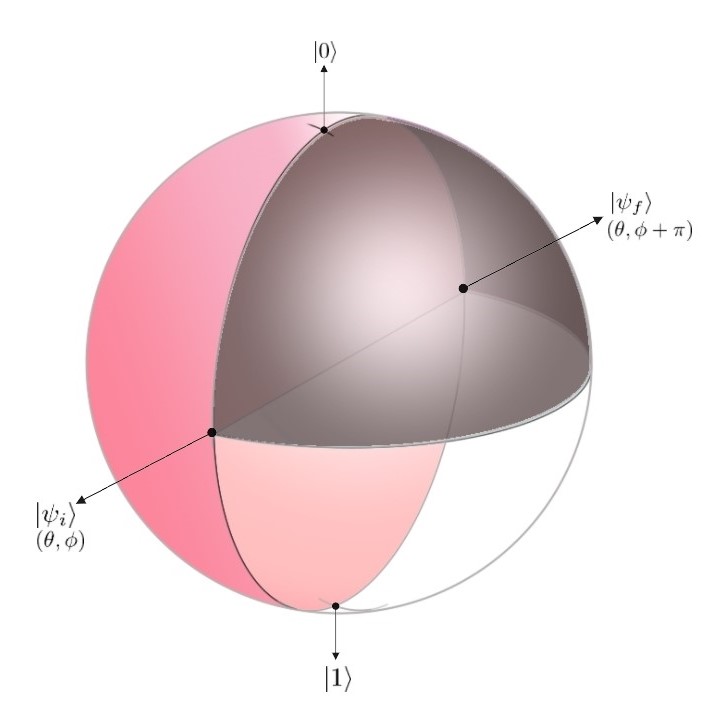}
	\caption{\textbf{Geometric phases associated with weak value amplification in 2-d Bloch sphere.}The grey region is a lune in the 2-d Bloch sphere with dihedral angle $\theta$; so it has area $2\theta$. This region is subtended by the state transition $\ket{\psi_i}\rightarrow \ket{0}\rightarrow \ket{\psi_f} \rightarrow \ket{\psi_i}$ which creates a geometric phase factor $\text{exp}(-i\theta)$. The pink region is a half sphere with area $2\pi$. A state transition $\ket{\psi_i}\xrightarrow[]{\text{via} \; \ket{0}} \ket{1}\xrightarrow[]{\text{via} \; \ket{0}}\ket{\psi_f}\xrightarrow[]{}\ket{\psi_i}$ subtends the grey and pink region together in Bloch sphere creating a geometric phase factor $-\text{exp}(-i\theta)$. These two state transitions generates the optimal weak value amplification when a spin 1/2 system is prepared in $\ket{\psi_i}$ with latitude and longitude $(\theta,\phi)$ in 2-d Bloch sphere and the system observable is $\sigma_z$.}
	\label{fig:geom}
\end{figure}

We take a spin-1/2 system as our reference for this example since the one qubit Bloch sphere is the simplest state-space geometry to visualize. The $\hat{\sigma}_z$ operator is taken as the system observable. Now we prepare the system in a state $\ket{\psi_i}=\text{cos}(\theta/2)\ket{0}+e^{i\phi}\text{sin}(\theta/2)\ket{1}$ which is represented as a point with latitude and longitude $(\theta,\phi)$ in 2-d Bloch sphere. The north and south poles refer to the $\ket{0}$ and $\ket{1}$ states, respectively. 

Now, the optimal efficiency in weak value is attained when the system is postselected in $\ket{\psi_f}=\text{cos}(\theta/2)\ket{0}-e^{i\phi}\text{sin}(\theta/2)\ket{1}$ which specifically refers to the point in Bloch sphere with coordinates $(\theta,\phi+\pi)$. Since in the optimal limit, the geometric phases $\Phi_g^{a_k}$ depend on the eigenvalues of the observable via sign functions, we can expect the following in the optimal limit,
\begin{align} 
\label{geometriceq1}
 \text{exp}[i \;\Phi_{\text{g}}^{1}\Big(\ket{\psi_i}_{\theta,\phi},\ket{0},\ket{\psi_f}_{\theta,\phi+\pi}\Big)] &=  \chi \\ 
 \text{exp}[i \;\Phi_{\text{g}}^{-1}\Big(\ket{\psi_i}_{\theta,\phi},\ket{1},\ket{\psi_f}_{\theta,\phi+\pi}\Big)] &=  -\chi,
\end{align}
where $\chi$ is some constant independent of the eigenvalues of the observable. The grey region in Figure \ref{fig:geom} is subtended by the transition path $\ket{\psi_i}\rightarrow \ket{0}\rightarrow \ket{\psi_f} \rightarrow \ket{\psi_i}$. All the transitions are executed through great circles in 2-d Bloch sphere.  The area of this region is $2\theta$, so this region will create a geometric phase factor $\text{exp}(-i\theta)$. This geometric phase corresponds to the first state transition in Eq.~(\ref{geometriceq1}). We take the second transition $\ket{\psi_i}\xrightarrow[]{\text{via} \; \ket{0}} \ket{1}\xrightarrow[]{\text{via} \; \ket{0}}\ket{\psi_f}\xrightarrow[]{}\ket{\psi_i}$. This transition subtends the black and the pink region together. The pink region is a half sphere, so it has area $2\pi$. Thus, the total area covered by the second transition is $2\pi+2\theta$. This creates a total geometric phase factor $\text{exp}(-i(2\pi+2\theta)/2)$ which equals $-\text{exp}(-i\theta)$. Thus taking $\chi = \text{exp}(-i\theta)$ reconciles this geometric understanding of weak value optimization with standard weak value optimization procedure.

\section{ Relation between postselected Fisher information and weak values} \label{App:QuasDer}
\color{black}The postselected quantum Fisher information is given by
\begin{equation}
\label{Eq:FishQpsShort2}
\mathcal{I}_{Q}(\theta\vert\psi_{\theta}^{ps}) = 4\langle\dot{\Psi}_{\theta}^{ps}\vert \dot{\Psi}_{\theta}^{ps}\rangle\frac{1}{p_{\theta}^{ps}}-4\vert\langle\dot{\Psi}_{\theta}^{ps}\vert\Psi_{\theta}^{ps}\rangle\vert^{2}\frac{1}{(p_{\theta}^{ps})^{2}},
\end{equation}
where unnormalized postselected quantum state is $\ket{\Psi_{\theta}^{ps}} = \hat{F} \hat{U}(\theta) \ket{\psi_0} $, where $\ket{\psi_0}\bra{\psi_0} \equiv \hat{\rho}_0$. $p_{\theta}^{ps} = \mathrm{Tr}(\hat{F}\hat{\rho}_{\theta})$ is the probability of postselection. In this supplementary note, we show that how the postselected Fisher information and weak values are related. Before doing this, we quickly recap the way of re-expressing Eq.~(\ref{Eq:FishQpsShort2}) in terms of the operator form. The first term of Eq.~(\ref{Eq:FishQpsShort2}) is 
\begin{align}
 \frac{4}{p_{\theta}^{ps}} \langle\dot{\Psi}_{\theta}^{ps}\vert \dot{\Psi}_{\theta}^{ps}\rangle
& =  \frac{4}{p_{\theta}^{ps}} \mathrm{Tr}\big( \hat{F}\dot{\hat{U}}(\theta)  \hat{\rho}_0 \dot{\hat{U}}^{\dagger}(\theta)  \hat{F}^{\dagger} \big) =  \frac{4}{p_{\theta}^{ps}} \mathrm{Tr}\Big( \hat{F} \hat{A} \hat{\rho}_{\theta} \hat{A} \Big)
 \label{Eq:AppTr1}
 \end{align}
We can express $\hat{A}$ and $\hat{F}$ in their corresponding eigendecompositions to rewrite Eq.~(\ref{Eq:AppTr1}) in terms of the doubly extended Kirkwood-Dirac quasiprobability distribution. Similarly, the second term of Eq.~(\ref{Eq:FishQpsShort2}) is
\begin{align}
 \frac{4}{(p_{\theta}^{ps})^2} \vert\langle\dot{\Psi}_{\theta}^{ps}\vert\Psi_{\theta}^{ps}\rangle\vert^{2}
=  \frac{4}{(p_{\theta}^{{ps}})^2} \big\vert  \mathrm{Tr}\big(  \hat{F} \hat{\rho}_{\theta} \hat{A} \big)\big\vert^2
\end{align} 
Combining the expressions above we obtain:
\begin{equation}
\mathcal{I}_Q(\theta \vert \psi_{\theta}^{ps}) = \frac{4}{p_{\theta}^{ps}} \mathrm{Tr}\Big( \hat{F} \hat{A} \hat{U}(\theta) \hat{\rho}_{0} \hat{U}(\theta)^{\dagger} \hat{A} \Big) - 
\frac{4}{(p_{\theta}^{ps})^2} \Big\vert  \mathrm{Tr}\Big(  \hat{F} \hat{U}(\theta) \hat{\rho}_{0} \hat{U}(\theta)^{\dagger} \hat{A} \Big)\Big\vert^2 .
%\end{split}
\end{equation}
Rearranging Eq.~(7) in the main text along with the replacement of total probability $p_{\theta}^{ps}$ as a sum of individual probabilities when the postselction is done with a rank-2 projector,
\begin{equation}\label{ghosh}
    %\begin{split}
        \mathcal{I}_Q(\theta \vert \psi_{\theta}^{ps})(p_{\theta}^{ps})^2=4\Big[(p_{\theta1}^{ps}+p_{\theta2}^{ps})\Big(\vert\langle{\psi_{\theta}}\vert\hat{A}\vert{f_{1}}\rangle\vert^{2}+\\ 
      ~  \vert\langle{\psi_{\theta}}\vert\hat{A}\vert{f_{2}}\rangle\vert^{2}\Big)-\vert\langle{\psi_{\theta}}\vert\hat{A}\vert{f_{1}}\rangle\\
      ~\langle f_{1}\vert\psi_{\theta}\rangle +
      \langle{\psi_{\theta}}\vert\hat{A}\vert{f_{2}}\rangle\langle f_{2}\vert\psi_{\theta}\rangle\vert^{2}\Big].
    %\end{split}
\end{equation}
Simplifying Eq.~(\ref{ghosh}) we find:
\begin{equation}
    \begin{split}
       \mathcal{I}_Q(\theta \vert \psi_{\theta}^{ps})(p_{\theta}^{ps})^2=4\Big[p_{\theta 1}^{ps}\vert\langle{\psi_{\theta}}\vert\hat{A}\vert{f_{1}}\rangle\vert^{2}+p_{\theta 2}^{ps}\vert\langle{\psi_{\theta}}\vert\hat{A}\vert{f_{2}}\rangle\vert^{2}+p_{\theta 2}^{ps}\vert\langle{\psi_{\theta}}\vert\hat{A}\vert{f_{1}}\rangle\vert^{2}+p_{\theta 1}^{ps}\vert\langle{\psi_{\theta}}\vert\hat{A}\vert{f_{2}}\rangle\vert^{2}\\
       ~-p_{\theta 1}^{ps}\vert\langle{\psi_{\theta}}\vert\hat{A}\vert{f_{1}}\rangle\vert^{2}-p_{\theta 2}^{ps}\vert\langle{\psi_{\theta}}\vert\hat{A}\vert{f_{2}}\rangle\vert^{2}-\langle\psi_{\theta}\vert\hat{A}\vert f_{1}\rangle\langle f_{1}\vert\psi_{\theta}\rangle\langle\psi_{\theta}\vert f_{2}\rangle\langle f_{2}\vert\hat{A}\vert\psi_{\theta}\rangle\\
       ~-\langle\psi_{\theta}\vert\hat{A}\vert f_{2}\rangle\langle f_{2}\vert\psi_{\theta}\rangle\langle\psi_{\theta}\vert f_{1}\rangle\langle f_{1}\vert\hat{A}\vert\psi_{\theta}\rangle\Big].
    \end{split}
\end{equation}
This expression can be written in a more concise way
\begin{equation}
         \mathcal{I}_Q(\theta \vert \psi_{\theta}^{ps})(p_{\theta}^{ps})^2=4\Big[p_{\theta 1}^{ps}p_{\theta 2}^{ps}\vert A_{w}^{1}\vert^2+p_{\theta 1}^{ps}p_{\theta 2}^{ps}\vert A_{w}^{2}\vert^2-p_{\theta 1}^{ps}p_{\theta 2}^{ps}\\
        ~A_{w}^{2}A_{w}^{1*}-p_{\theta 1}^{ps}p_{\theta 2}^{ps}A_{w}^{1}A_{w}^{2*}\Big],
\end{equation}
where $A_w^{k}=\frac{\langle\psi_\theta\vert\hat{A}\vert f_k\rangle}{\langle\psi_\theta\vert f_k\rangle}$. Further simplification leads to
\begin{equation}
      \mathcal{I}_Q(\theta \vert \psi_{\theta}^{ps})(p_{\theta}^{ps})^2=4p_{\theta 1}^{ps}p_{\theta 2}^{ps}\vert A_{w}^{1}-A_{w}^{2}\vert^2.
\end{equation}
Repeating the same calculation for rank-3 projector, we obtain:
\begin{equation}
      \mathcal{I}_Q(\theta \vert \psi_{\theta}^{ps})(p_{\theta}^{ps})^2=4\Big[p_{\theta 1}^{ps}p_{\theta 2}^{ps}\vert A_{w}^{1}-A_{w}^{2}\vert^2+p_{\theta 2}^{ps}p_{\theta 3}^{ps}\vert A_{w}^{2}-A_{w}^{3}\vert^2+ \\
      ~p_{\theta 3}^{ps}p_{\theta 1}^{ps}\vert A_{w}^{3}-A_{w}^{1}\vert^2\Big].
\end{equation}
From these two expressions we can easily extrapolate a generalized formula for rank-$n$ projector
\begin{equation}
      \mathcal{I}_Q(\theta \vert \psi_{\theta}^{ps})(p_{\theta}^{ps})^2=4\sum_{m<n}p_{\theta m}^{ps}p_{\theta n}^{ps}\vert A_w^{m}-A_w^{n}\vert^2.
\end{equation}

\section{Efficiency of the postselected metrology}
We begin by writing the the evolved pure state in the form as follows,
\begin{equation}
    \hat{U}(\theta)\hat{\rho}_0\hat{U}(\theta)^\dagger=\ket{\psi_\theta}\bra{\psi_\theta}.
\end{equation}
Rewriting the first part of RHS in Eq.~(7) in the main text followed by a decomposition for the postselection projectors yield,
\begin{equation}
    \mathrm{Tr}\Big( \hat{F} \hat{A} \hat{U}(\theta) \hat{\rho}_{0} \hat{U}(\theta)^{\dagger} \hat{A} \Big)= \sum_{\substack{f_k\in \mathcal{F}^{ps}}} \vert \langle \psi_\theta \vert \hat{A}\vert f_k \rangle \vert^2 .
\end{equation}
Using weak value optimization and rewriting the previous equation we get,
\begin{equation}
    \mathrm{Tr}\Big( \hat{F} \hat{A} \hat{U}(\theta) \hat{\rho}_{0} \hat{U}(\theta)^{\dagger} \hat{A} \Big) = \sum_{\substack{f_k\in \mathcal{F}^{ps}}} p^{ps}_{\theta k}\vert A_{w}^k\vert^2 .
\end{equation}
This expression establishes a connection between the efficiency of weak value amplification and the efficiency of postselected quantum metrology. Using similar method we can derive the following regarding the second term in RHS of Eq.~(39) in the main text
\begin{equation}
    \mathrm{Tr}\Big(  \hat{F} \hat{U}(\theta) \hat{\rho}_{0} \hat{U}(\theta)^{\dagger} \hat{A} \Big) = \sum_{{f_k}\in \mathcal{F}^{ps} } \langle \psi_\theta \vert \hat{A} \vert f_k\rangle \langle f_k \vert \psi_\theta \rangle .
\end{equation}
Rewriting the same equation after a spectral decomposition of $\hat{A}$ we obtain,
\begin{equation}
    \mathrm{Tr}\Big(\hat{F} \hat{U}(\theta) \hat{\rho}_{0} \hat{U}(\theta)^{\dagger} \hat{A} \Big) = \sum_m \sum_{{f_k}\in \mathcal{F}^{ps} } a_m \langle \psi_\theta \vert a_m \rangle \langle a_m \vert f_k\rangle \langle f_k \vert \psi_\theta \rangle = \sum_m \sum_{{f_k}\in \mathcal{F}^{ps} } a_m q_{\ket{\psi_\theta}}^{(a_m,f_k)}.
\end{equation}
Putting these expressions altogether, we can easily obtain Eq.~(39) in the main text.

\section{Positivity of standard and doubly extended Kikwood-Dirac distribution}
In this section we discuss the correspondence between KD and doubly extended KD positivity for pure states. The doubly extended KD distribution for some pure state $\ket{\psi}$ is defined for three arbitrary orthonormal set of projectors $\{\hat{A_j}\}$, $\{\hat{B_k}\}$ and $\{\hat{C_l}\}$ where $j,k,l$ belong to the same index set,
\begin{equation}
    q^{\ket{\psi}}_{a_j,b_k,c_l} = Tr(\ket{\psi}\bra{\psi}\hat{A_j}\hat{C_l}\hat{B_k}).
\end{equation}
Similar to the KD distribution this distribution can be compared with four vertex Bergman invariant. 

In Eq.~(6) in the main text, first and second arguments belong to the same basis set. Rewriting the distribution taking $\hat{C_k}$ as the postselection projectors yields,
\begin{equation}
    q^{\ket{\psi}}_{a_j,a_l,f_k} = \langle \psi \vert a_j \rangle \langle a_j\vert f_k \rangle \langle f_k \vert a_l \rangle \langle a_l \vert \psi \rangle .
\end{equation}
Suppose it is given that our KD distribution $q^{\ket{\psi}}_{a_m,f_k}:=\langle \psi \vert a_m \rangle \langle a_m\vert f_k \rangle \langle f_k \vert \psi \rangle$ is positive everywhere. We rewrite the previous equation after multiplying $\vert  \langle f_k \vert \psi \rangle  \vert^2$ in the numerator and the denominator as follows,
\begin{equation}
    \begin{split}
        q^{\ket{\psi}}_{a_j,f_k,a_l} = \frac{\langle \psi \vert a_j \rangle \langle a_j\vert f_k \rangle \langle f_k\vert \psi \rangle \langle \psi \vert f_k \rangle \langle f_k \vert a_l \rangle \langle a_l \vert \psi \rangle}{\vert  \langle f_k \vert \psi \rangle  \vert^2} 
        =\frac{q^{\ket{\psi}}_{a_j,f_k}(q^{\ket{\psi}}_{a_l,f_k})^*}{\vert  \langle f_k \vert \psi \rangle  \vert^2}.
    \end{split}
\end{equation}
From these equations, we can explicitly see that all the terms in this doubly extended KD distribution are positive, given that the corresponding state is pure and the KD distribution is positive.

\section{Derivation of the relationship between efficiency of WVA and geometric phase}
We start with the expression for geometric phase,
\begin{equation}
\label{e1}
    \begin{split}
        \Phi_{\text{g}}^{a_k}(\ket{\psi_i},\ket{a_k},\ket{\psi_f}) = \text{arg}(\langle \psi_f \vert a \rangle \langle a \vert \psi_i \rangle \langle \psi_i \vert \psi_f \rangle) \\
        =\text{arg}(\langle \psi_f \vert a \rangle \langle a \vert \psi_i \rangle) + \text{arg}(\langle \psi_i \vert \psi_f \rangle).
    \end{split}
\end{equation}
Rearranging Eq.~(\ref{e1}), we obtain,
\begin{equation}
\label{e2}
    \text{arg}(\langle \psi_f \vert a \rangle \langle a \vert \psi_i \rangle) = \Phi_{\text{g}}^{a_k}(\ket{\psi_i},\ket{a_k},\ket{\psi_f}) - \text{arg}(\langle \psi_i \vert \psi_f \rangle) .
\end{equation}
 Rewriting the efficiency of WVA in terms of geometric phase,
\begin{equation}
\label{e3}
\begin{split}
    \eta(p_s,A_w) & = \big\vert \langle \psi_f \vert \hat{A} \vert \psi_i\rangle \big\vert^2 \\
    & = \Big\vert \sum_k a_k \langle \psi_f \vert a_k \rangle \langle a_k \vert \psi_i\rangle \Big\vert^2 \\
    & = \Big\vert \sum_k a_k \big\vert \langle \psi_f \vert a_k \rangle \langle a_k \vert \psi_i \rangle \big\vert \text{exp}(i \; \text{arg}(\langle \psi_f \vert a_k \rangle \langle a_k \vert \psi_i \rangle)) \Big\vert^2.
\end{split}
\end{equation}
Replacing Eq.~(\ref{e2})in Eq.~(\ref{e3}), we get,
\begin{equation}
    \eta(p_s,A_w)= \Big\vert \sum_k a_k \big\vert \langle \psi_f \vert a_k \rangle \langle a_k \vert \psi_i \rangle \big\vert \text{exp}(i\; [\Phi_{\text{g}}^{a_k}(\ket{\psi_i},\ket{a_k},\ket{\psi_f})-\text{arg}(\langle \psi_i \vert \psi_f \rangle)]) \Big\vert^2.
\end{equation}
Since $\langle \psi_i \vert \psi_f \rangle$ is not dependent on the eigenvalues of the system observable it creates a global phase and does not contribute in the efficiency. Thus we obtain,
\begin{equation}
    \eta(A_w,p_s) =\Big\vert\sum_{k} a_{k} \big\vert\langle\psi_{f}\vert a_{k}\rangle\langle a_{k} \vert\psi_{i}\rangle\big\vert\text{exp}(i\; \Phi_{g}^{a_k})\Big\vert^{2}.
\end{equation}

\end{document}